\documentclass[letterpaper,11pt]{article}
\usepackage[margin=1in]{geometry}
\usepackage{amsmath, amsfonts, amssymb, amsthm}
\usepackage{thm-restate}
\usepackage{bbm,url}
\usepackage[colorlinks=true,urlcolor=black,linkcolor=black,citecolor=black]{hyperref}
\usepackage{graphicx}
\usepackage{color}
\usepackage{subcaption}
\usepackage{hhline}
\usepackage{pifont}
\usepackage{multirow}
\usepackage{multicol}

\usepackage{nicefrac,bm}
\usepackage{algorithm}
\usepackage[noend]{algpseudocode}
\usepackage{thmtools}
\usepackage{thm-restate}

\newcommand\NoDo{\renewcommand\algorithmicdo{}}
\newcommand\ReDo{\renewcommand\algorithmicdo{\textbf{do}}}

\allowdisplaybreaks

\newtheorem{theorem}{Theorem}
\newtheorem*{theorem*}{Theorem}

\newtheorem{lemma}{Lemma}
\newtheorem{claim}{Claim}
\newtheorem{result}{Result}

\theoremstyle{definition}
\newtheorem{definition}{Definition}

\newcommand{\hide}[1]{}

\newcommand{\s}[1]{\mathsf{#1}}
\newcommand{\mBB}{\textrm{mBB}}

\newcommand{\x}{\mathbf{x}}
\newcommand{\y}{\mathbf{y}}

\newcommand{\p}{\mathbf{p}}

\newcommand{\N}{\mathbb N}

\newcommand{\R}{\mathbb R}
\newcommand{\Rp}{\mathbb{R}^+}
\newcommand{\poly}[1]{\mathsf{poly}(#1)}

\title{Fair and Efficient Allocations of Chores under Bivalued Preferences\thanks{Work on this paper is supported by NSF Grant CCF-1942321 (CAREER)}} 

\author{Jugal Garg\footnote{University of Illinois at Urbana-Champaign, USA} \\
\texttt{\small jugal@illinois.edu} 
\and
Aniket Murhekar\footnote{University of Illinois at Urbana-Champaign, USA}\\
\texttt{\small aniket2@illinois.edu}
\and
John Qin\footnote{University of Illinois at Urbana-Champaign, USA}\\
\texttt{\small johnqin2@illinois.edu}
}

\date{}
\begin{document}
\renewcommand{\arraystretch}{1.2}
\maketitle

\begin{abstract}
We study the problem of fair and efficient allocation of a set of indivisible \textit{chores} to agents with additive cost functions. We consider the popular fairness notion of envy-freeness up to one good (EF1) with the efficiency notion of Pareto-optimality (PO). While it is known that an EF1+PO allocation exists and can be computed in pseudo-polynomial time in the case of goods, the same problem is open for chores.

Our first result is a strongly polynomial-time algorithm for computing an EF1+PO allocation for \textit{bivalued instances}, where agents have (at most) two disutility values for the chores. To the best of our knowledge, this is the first non-trivial class of indivisible chores to admit an EF1+PO allocation and an efficient algorithm for its computation. 

We also study the problem of computing an envy-free (EF) and PO allocation for the case of divisible chores. While the existence of an EF+PO allocation is known via competitive equilibrium with equal incomes, its efficient computation is open. Our second result shows that for bivalued instances, an EF+PO allocation can be computed in strongly polynomial-time.
\end{abstract}

\section{Introduction}
The problem of fair division is concerned with allocating items to agents in a \textit{fair} and \textit{efficient} manner. Formally introduced by Steinhaus \cite{steinhaus}, fair division is an active area of research studied across fields like computer science and economics. Most work has focused on the fair division of \textit{goods}: items which provide non-negative \textit{value} (or utility) to the agents to whom they are allocated. However, several practical scenarios involve \textit{chores} (or bads). Chores are items which impose a \textit{cost} (or disutility) to the agent to whom they are allocated. For instance, household chores such as cleaning and cooking often need to be fairly distributed among members of the household. Likewise, teachers have to divide teaching load, stakeholders have to divide liabilities upon dissolution of a firm, etc. These examples highlight the importance of allocating \textit{chores} in a fair and efficient manner. Agencies responsible for designing such allocations must take into account the differences in preferences of agents in order for the allocation to be acceptable to all those involved. 

Arguably, the most popular notion of fairness is envy-freeness (EF) \cite{foleyEF, VARIAN1974-efpo}, which requires that every agent weakly prefers the bundle of items allocated to them over the bundle of any other agent. When items are \textit{divisible}, i.e., can be shared among agents, EF allocations are known to exist. However, in the case of \textit{indivisible} items, EF allocations need not exist. For instance, while dividing one chore between two agents, the agent who is assigned the chore will envy the other. Since the fair division of indivisible items remains an important problem, several relaxations of envy-freeness have been defined, first in the context of goods, and later adapted to chores. 

A widely studied relaxation of envy-freeness is \textit{envy-freeness up to one item} (EF1), defined by Budish~\cite{budish2011approxCEEI} in the context of goods. For chores, an allocation is said to be EF1 if for every agent, the envy disappears after removing one chore assigned to her. It is known that an EF1 allocation of chores exists and can be efficiently computed~\cite{lipton,bhaskar2020chores}. However, an EF1 allocation may be highly inefficient. Consider for example two agents $A_1$ and $A_2$ and two chores $j_1$ and $j_2$ where $A_i$ has almost zero cost for $j_i$ and high cost for the other chore. The allocation in which $j_1$ is assigned to $A_2$ and $j_2$ is assigned to $A_1$ is clearly EF1. However both agents incur high cost, which is highly inefficient. The allocation in which $j_i$ is assigned to $A_i$ is more desirable since it is both fair as well as efficient.

The standard notion of economic efficiency is Pareto optimality (PO). An allocation is said to be PO if no other allocation makes an agent better off without making someone else worse off. Fractional Pareto optimality (fPO) is a stronger notion requiring that no other fractional allocation makes an agent better off without making someone else worse off. Every fPO allocation is therefore PO, but not vice-versa.

An important question is whether the fairness and efficiency notions of EF1 and PO (or fPO) can be achieved in conjunction, and if so, can they be computed in polynomial-time. For the case of goods, Barman et al.~\cite{Barman18FFEA} showed that an EF1+PO allocation exists and can be computed in pseudo-polynomial time. Improving this result, Garg and Murhekar~\cite{murhekar2021aaai} showed that an EF1+fPO allocation can be computed in pseudo-polynomial time. For the case of chores, it is unclear whether EF1+PO allocations even exist, except for simple cases like identical valuations. Settling the existence of EF1+PO allocations (and developing algorithms for computing them) has turned out to be a challenging open problem.

In this paper, we present the first non-trivial results on the EF1+PO problem for chores. We study the class of \textit{bivalued} instances, where there are only two costs, i.e., for every agent a chore costs either $a$ or $b$, for two non-negative numbers\footnote{\label{footnote:binary}In fact, we can assume the two values are positive, since one of them being zero implies the setting is \textit{binary}, in which case computing an EF1+PO allocation is trivial by first assigning chores to agents which have 0 cost for them, and then allocating almost equal number of chores of cost 1 to everyone.} $a$ and $b$. Bivalued instances are a well-studied class in the fair division literature, we list several works in Section~\ref{sec:relatedwork}. In particular, Amanatidis et al.~\cite{amanatidis2020mnwefx} showed that allocations that are envy-free upto \textit{any} good (EFX), which generalizes EF1, can be efficiently computed for bivalued goods. Recently, Garg and Murhekar~\cite{garg2021sagt} showed that EFX+fPO allocations can be computed in polynomial-time for bivalued goods. Showing positive results for bivalued chores, our first result is:

\begin{result}
For bivalued instances with $n$ agents and $m$ indivisible chores, an EF1+fPO allocation exists and can be computed in $\poly{n,m}$-time.
\end{result}

Next, we study the problem of computing an EF+PO allocation of divisible chores. For goods, it is known that an EF+PO allocation always exists~\cite{VARIAN1974-efpo} and is in fact polynomial-time computable via the Eisenberg-Gale convex program~\cite{nisan2007agtbook}. This is done by computing the competitive equilibrium with equal incomes (CEEI). Here, the idea is to provide each agent with the same amount of fictitious money, and then find prices and an allocation of items such that all items are completely bought and each agent buys her most preferred bundle subject to her budget constraint. This is an example of a \textit{market} where demand (of agents) equals supply (of items), and is known as the Fisher market. For goods, there are polynomial-time algorithms for computing the competitive equilibrium (CE)~\cite{devanur2008mkteqb,orlin2010fisher,Vegh16}. For chores, the problem is harder: Bogomolnaia et al.~ \cite{bogomolnaia2017manna} showed that the CE rule can be non-convex, multi-valued and disconnected. Algorithms with exponential run-times are known for computing CE for chores~\cite{branzei2019choresceei,GargM20,chaudhury2020manna,GargHMS21}, but designing a polynomial-time algorithm is an open problem. Working towards this goal, our second result shows:  

\begin{result}
For bivalued instances with $n$ agents and $m$ divisible chores, an EF+fPO allocation can be computed in $\poly{n,m}$-time.
\end{result}

\subsection{Further related work}\label{sec:relatedwork}

Barman et al.~\cite{Barman18FFEA} showed that for $n$ agents and $m$ goods, an EF1+PO allocation can be computed in time $\poly{n,m,V}$, where $V$ is the maximum utility value. Their algorithm first perturbs the values to a desirable form, and then computes an EF1+fPO allocation for the perturbed instance, which for a small-enough perturbation is EF1+PO for the original instance. Their approach is via \textit{integral market equilibria}, which guarantees fPO at every step, and the concept of price-envy-freeness up to one good (pEF1) which is a strengthening of EF1. Using similar tools, Garg and Murhekar~\cite{murhekar2021aaai} showed that an EF1+fPO allocation can be computed in $\poly{n,m,V}$-time. They also showed that an EF1+fPO allocation can be computed in $\poly{n,m}$-time for $k$-ary instances (agents have at most $k$ values for the goods) where $k$ is a constant. It may seem a natural idea to try and use these approaches for chores, however they do not extend easily. While our algorithm also uses integral market equilibria to obtain the fPO property and pEF1 for chores to argue about EF1, our algorithm and its analysis is much more involved and significantly different from previous works.

Bivalued preferences are a well-studied class in literature. The following results are for the goods setting. Aziz et al.~\cite{aziz2019PO} showed PO is efficiently verifiable for bivalued instances and coNP-hard for 3-valued instances; Aziz~\cite{aziz2020hyllandzeckhauser}, and Vazirani and  Yannakakis~\cite{vazirani2020} studied the Hylland-Zeckhauser scheme for probabilistic assignment of goods in bivalued instances; and Bogomolnaia and Moulin~\cite{bogomolnaia04dichotomous} studied matching problems with bivalued (dichotomous) preferences. More generally, instances with few values have also been studied: Barman et al.~\cite{barman2018binarynsw} showed that EF1+PO allocations can be computed for binary valuations; Babaioff et al.~\cite{babaioff21dichotomous} studied truthful mechanisms for dichotomous valuations; Golovin~\cite{Golovin2005MaxminFA} presented approximation algorithms and hardness results for computing max-min fair allocations in 3-valued instances;Bliem et al.~\cite{bliem2016param} studied fixed-parameter tractability for computing EF+PO allocations with parameter $n+z$, where $z$ is the number of values; and Garg et al.~\cite{garg09papers} studied leximin assignments of papers ranked by reviewers on a small scale, in particular they present an efficient algorithm for 2 ranks, i.e., ``high or low interest'' and show NP-hardness for 3 ranks. Such instances have also been studied in resource allocation contexts, including makespan minimization with 2 or 3 job sizes~\cite{woeginger,chakrabarty-makespan}.

The fairness notion of equitability requires that each agent get the same amount of utility or disutility.  Similar to EF1 and EFX, equitability up to one (resp. any) item (EQ1 (resp. EQX)) are relaxations of equitability. Using approaches inspired by \cite{Barman18FFEA}, pseudo-polynomial time algorithms for computing EQ1+PO allocations were developed for both goods~\cite{freeman2019eqxpo} and chores~\cite{freeman2020chores}. For bivalued goods, an EQX+PO allocation is polynomial time computable \cite{garg2021sagt}.

\section{Preliminaries}\label{sec:prelim}

\paragraph{Problem instance.} 
A \textit{fair division instance} (of chores) is a tuple $(N,M,C)$, where $N = [n]$ is a set of $n\in\N$ agents, $M = [m]$ is a set of $m\in\N$ indivisible chores, and $C = \{c_1,\dots,c_n\}$ is a set of cost or disutility functions, one for each agent $i\in N$. Each cost function $c_i : M \rightarrow \mathbb{R}_{\ge 0}$ is specified by $m$ numbers $c_{ij}\in\R_{\ge0}$, one for each chore $j\in M$, which denotes the cost agent $i$ has for performing (receiving) chore $j$. We assume that the cost functions are additive, that is, for every agent $i \in N$, and for $S \subseteq M$, $c_i(S) = \sum_{j\in S} c_{ij}$. For notational ease, we write $c(S\setminus j)$ instead of $c(S\setminus \{j\})$.

We call a fair division instance $(N, M, C)$ a \textit{bivalued instance} if there exist $a,b\in \R_{\ge 0}$, with $a\ge b$, such that for all $i\in N$ and $j\in M$, $c_{ij} \in \{a,b\}$. That is, the cost of any chore to any agent is one of at most two given numbers. By scaling the costs, we can assume w.l.o.g. for bivalued instances that all costs $c_{ij} \in \{1,k\}$, where $k = a/b \ge 1$. Such a scaling is w.l.o.g., since our fairness and efficiency properties are scale-invariant. Given such an instance, we partition the set of chores into sets of low-cost chores $M_{low}$ and high-cost chores $M_{high}$:
\begin{itemize}
    \item $M_{low} = \{j\in M: \exists  i\in N \, \text{ s.t. } c_{ij} = 1 \}$, and
    \item $M_{high} = \{j\in M:\forall i\in N, \, c_{ij} = k \}$.
\end{itemize}
Further, we assume that both $a,b > 0$, since if $b=0$ then re-scaling the values transforms the instance into the simpler binary case, for which efficient algorithms are known (Footnote~\ref{footnote:binary}). We can additionally assume for bivalued instances that for every agent $i$, there is at least one chore $j$ s.t. $c_{ij} = 1$. This is w.l.o.g., since if $c_{ij}=k$ for all $j\in M$, then we can re-scale costs to set $c_{ij} = 1$ for all $j\in M$.
    
\paragraph{Allocation.} An \textit{allocation} $\x$ of chores to agents is an $n$-partition $(\x_1, \dots, \x_n)$ of the chores, where agent $i$ is allotted $\x_i \subseteq M$ and gets a total cost of $c_i(\x_i)$. A \textit{fractional allocation} $\x \in [0,1]^{n\times m}$ is a fractional assignment such that for each chore $j\in M$, $\sum_{i\in N} x_{ij} = 1$. Here, $x_{ij}\in[0,1]$ denotes the fraction of chore $j$ allotted to agent $i$.

\paragraph{Fairness notions.} An allocation $\x$ is said to be \textit{envy-free up to one chore} (EF1) if for all $i,h \in N$, there exists a chore $j\in \x_i$ such that $c_i(\x_i\setminus j) \le c_i(\x_h)$.

We say that an agent $i$ \textit{EF1-envies} an agent $h$ if for all $j\in \x_i$, $c_i(\x_i \setminus j) > c_i(\x_h)$, i.e., the EF1 condition between $i$ and $h$ is violated. 

A (fractional) allocation $\x$ is said to be \textit{envy-free} if for all $i,h\in N$, $c_i(\x_i)\le c_i(\x_h)$. We say that an agent $i$ \textit{envies} an agent $h$ if $c_i(\x_i) > c_i(\x_h)$, i.e., the EF condition between $i$ and $h$ is violated. 

\paragraph{Pareto-optimality.} An allocation $\y$ dominates an allocation $\x$ if $c_i(\y_i) \le c_i(\x_i), \forall i$ and there exists $h$ s.t. $c_h(\y_h) < c_h(\x_h)$. An allocation is said to be \textit{Pareto optimal} (PO) if no allocation dominates it. Further, an allocation is said to be fractionally PO (fPO) if no fractional allocation dominates it. Thus, a fPO allocation is PO, but not vice-versa.

\paragraph{Fisher markets.} A Fisher market or a \textit{market instance} is a tuple $(N,M,C,e)$, where the first three terms are interpreted as before, and $e=\{e_1,\dots,e_n\}$ is the set of agents' \textit{minimum payments}, where $e_i \ge 0$, for each $i\in N$. In this model, chores can be allocated fractionally. Given a payment vector, also called a \textit{price}\footnote{We refer to payments as prices for sake of similarity with the Fisher market model in the goods case.} vector, $\p=(p_1, \dots, p_m)$, each chore $j$ pays $p_j$ per unit of chore. Agents perform chores in exchange for payment. Given chore payments, each agent $i$ aims to obtain the set of chores that minimizes her total cost subject to her payment constraint, i.e., receiving a total payment of at least $e_i$. 

Given a (fractional) allocation $\x$ with a price vector $\p$, the \textit{spending}\footnote{This is actually the \textit{earning} of agent $i$, but we refer to earning as spending for sake of similarity with the Fisher market model in the goods case.} of an agent $i$ under $(\x,\p)$ is given by $\p(\x_i) = \sum_{j\in M} p_jx_{ij}$. We define the \textit{bang-per-buck} ratio $\alpha_{ij}$ of chore $j$ for an agent $i$ as $\alpha_{ij} = c_{ij}/p_j$, and the \textit{minimum bang-per-buck} (mBB) ratio as $\alpha_i = \min_{j} \alpha_{ij}$. We define $\mBB_i = \{j\in M : c_{ij}/p_j = \alpha_i \}$, called the \textit{mBB-set}, to be the set of chores that give mBB to agent $i$ at prices $\p$. We say $(\x,\p)$ is \textit{`on mBB'} if for all agents $i$ and chores $j$, $x_{ij} > 0  \Rightarrow j\in \mBB_i$. For integral $\x$, this means that $\x_i\subseteq \mBB_i$ for all $i\in N$. 

A \textit{market equilibrium} or \textit{market outcome} is a (fractional) allocation $\x$ of the chores to the agents and set of prices $\p$ of the chores satisfying the following properties:

\begin{itemize}
\item the market clears, i.e., all chores are fully allocated. Thus, for all $j$, $\sum_{i\in N} x_{ij} = 1$,
\item each agent receives their minimum payment, for all $i\in N$, $\p(\x_i) = \sum_{j\in M} x_{ij} p_j = e_i$, and,
\item agents only receive chores that give them minimum bang-per-buck, i.e., $(\x,\p)$ is on mBB. 
\end{itemize}

Given a market outcome $(\x,\p)$ with $\x$ integral, we say it is \textit{price envy-free up to one chore} (pEF1) if for all $i,h\in N$ there is a chore $j\in \x_i$ such that $\p(\x_i\setminus j) \le \p(\x_h)$. We say that an agent $i$ \textit{pEF1-envies} an agent $h$, if for all $j\in \x_i$, $\p(\x_i\setminus j) > \p(\x_h)$, i.e., the pEF1 condition between $i$ and $h$ is violated. For integral market outcomes on mBB, the pEF1 condition implies the EF1 condition. 

\begin{lemma}\label{lem:pEF1impliesEF1}
Let $(\x,\p)$ be an integral market outcome on mBB. If $(\x,\p)$ is pEF1 then $\x$ is EF1 and fPO.
\end{lemma}
\begin{proof}
We first show that $(\x,\p)$ forms a market equilibrium for the Fisher market instance $(N,M,C,e)$, where for every $i\in N$, $e_i = \p(\x_i)$. It is easy to see that the market clears and each agent receives their minimum payment. Further $\x$ is on mBB as per our assumption. Now the fact that $\x$ is fPO follows from the First Welfare Theorem \cite{mas1995microeconomic}, which shows that for any market equilibrium $(\x,\p)$, the allocation $\x$ is fPO. 

Since $(\x,\p)$ is pEF1, for all pairs of agents $i,h \in N$, there is some chore $j\in \x_i$ s.t. $\p(\x_i\setminus j)\le \p(\x_h)$. Since $(\x,\p)$ is on mBB, $\x_i \subseteq \mBB_i$. Let $\alpha_i$ be the mBB-ratio of $i$ at the prices $\p$. By the definition of mBB, $c_i(\x_i\setminus j) = \alpha_i \p(\x_i\setminus j)$, and $c_i(\x_h) \ge \alpha_i \p(\x_h)$. Combining these implies $\x$ is EF1.
\end{proof}

Our Algorithm~\ref{alg:ef1po} begins with and maintains an integral market outcome $(\x,\p)$ on mBB\footnote{Note that although our algorithm only maintains the allocation $\x$ and prices $\p$, the associated Fisher market instance is always implicitly present by setting $e_i = \p(\x_i)$ as in the proof of Lemma~\ref{lem:pEF1impliesEF1}}, and modifies $\x$ and $\p$ appropriately to eventually arrive at an outcome $(\x,\p)$ on mBB where the pEF1 condition is satisfied. Lemma~\ref{lem:pEF1impliesEF1} then ensures that $\x$ is EF1+fPO.

We now define least spenders as agents with minimum spending, and big spenders as agents with maximum spending after the removal of their highest-priced chore.
\begin{definition}[Least and big spenders]\label{def:lsbs}
An agent $\ell \in \s{argmin}_{i\in N} \p(\x_i)$ is referred to as a \textit{least spender} (LS). An agent $b \in \s{argmax}_{i\in N} \min_{j\in \x_i} \p(\x_i \setminus j)$ is referred to as a \textit{big spender} (BS). 
\end{definition}

We break ties arbitrarily to decide a unique LS and BS. Together with Lemma~\ref{lem:pEF1impliesEF1}, the following lemma shows that in order to obtain an EF1 allocation, it is sufficient to focus on the pEF1-envy the big spender has towards the least spender.

\begin{lemma}\label{lem:BSenviesLS}
Let $(\x,\p)$ be an integral market outcome on mBB. If $\x$ is not EF1, then the big spender $b$ pEF1-envies the least spender $\ell$. 
\end{lemma}
\begin{proof}
If $\x$ is not EF1, then Lemma~\ref{lem:pEF1impliesEF1} implies that $\x$ is not pEF1. Hence there is a pair of agents $i,h$ s.t. for every chore $j\in \x_i$, $\p(\x_i \setminus j) > \p(\x_h)$. By the definition of big spender, we know $\p(\x_b \setminus j') \ge \p(\x_i \setminus j)$, for every $j'\in \x_b$. By the definition of least spender, $\p(\x_i) \ge \p(\x_\ell)$. Putting these together we get $\p(\x_b \setminus j') > \p(\x_\ell)$ for every $j'\in \x_b$, implying that $b$ pEF1-envies $\ell$.
\end{proof}

Given a market outcome $(\x,\p)$ on mBB, we define the mBB graph to be a bipartite graph $G = (N,M,E)$ where for an agent $i$ and chore $j$, $(i,j)\in E$ iff $j\in \mBB_i$. Further, an edge $(i, j)$ is called an \textit{allocation edge} if $j \in \x_i$, otherwise it is called an \textit{mBB} edge. 

For agents $i_0,\dots,i_{\ell}$ and chores $j_1,\dots,j_{\ell}$, a path $P = (i_0, j_1, i_1, j_2, \dots, j_{\ell}, i_{\ell})$ in the mBB graph, where for all $1\le \ell' \le \ell$, $j_{\ell'} \in \x_{i_{\ell'-1}}\cap \mBB_{i_{\ell'}}$, is called a \textit{special} path. We define the \textit{level} $\lambda(h; i_0)$ of an agent $h$ w.r.t. $i_0$ to be half the length of the shortest special path from $i_0$ to $h$, and to be $n$ if no such path exists. A path $P = (i_0, j_1, i_1, j_2, \dots, j_{\ell}, i_{\ell})$ is an \textit{alternating path} if it is special, and if $\lambda(i_0; i_0) < \lambda(i_1; i_0) < \dots < \lambda(i_\ell; i_0)$, i.e. the path visits agents in increasing order of their level w.r.t. $i_0$. Further, the edges in an alternating path alternate between allocation edges and mBB edges. Typically, we consider alternating paths starting from a big spender agent.

\begin{definition}[Component $C_i$ of a big spender $i$]\label{def:component}
For a big spender $i$, define $C_i^{\ell}$ to be the set of all chores and agents which lie on alternating paths of length $\ell$ starting from $i$. Call $C_i = \bigcup_{\ell} C_i^{\ell}$ the \textit{component} of $i$, i.e., the set of all chores and agents reachable from $i$ through alternating paths.
\end{definition}

\section{EF1+fPO allocation of indivisible chores}\label{sec:indivisible}
In this section, we present our main result:
\begin{theorem}\label{thm:indivisible}
Given a bivalued fair division instance $(N,M,C)$ of indivisible chores with all $c_{ij} \in \{a,b\}$ for some $a,b\in \mathbb{R}^+$, an EF1+fPO allocation can be computed in strongly polynomial-time.
\end{theorem}

We prove Theorem~\ref{thm:indivisible} by showing that our Algorithm~\ref{alg:ef1po} computes an EF1+fPO allocation in polynomial-time. To ensure smooth presentation, we defer some proofs to Appendix~\ref{app:indivisible}. 

\subsection{Obtaining Initial Groups}\label{sec:makeinitgroups}
Recall that we can scale the costs so that they are in $\{1,k\}$, for some $k>1$. The first step of Algorithm~\ref{alg:ef1po} is to obtain a partition of the set $N$ of agents into \textit{groups} $N_1,\dots,N_R$ with desirable properties. For this, we use Algorithm~\ref{alg:MakeInitGroups} (called $\s{MakeInitGroups}$). 

\begin{algorithm}[ht]
\caption{$\s{MakeInitGroups}$}\label{alg:MakeInitGroups}
\textbf{Input:} Fair division instance $(N,M,C)$ with $c_{ij}\in \{1,k\}$\\
\textbf{Output:} Integral alloc. $\x$, prices $\p$, agent groups $\{N_r\}_{r\in [R]}$
\begin{algorithmic}[1]
\State $(\x,\p) \gets$ initial cost minimizing integral market allocation, where $p_j = c_{ij}$ for $j \in \x_i$.
\State $R\gets 1$, $N' \gets N$
\While{$N' \neq \emptyset$}
\State $b \gets \s{argmax}_{i \in N'} \min_{j\in \x_i} \p(\x_i\setminus j)$ \Comment{Big Spender}
\State $C_b \gets$ Component of $b$ \Comment{See Definition~\ref{def:component}}
\NoDo
\While{$\exists$ agent $i \in C_b$ s.t. $\forall j \in \x_b$, $\p(\x_b \setminus j) > \p(\x_i)$}
\ReDo
\State Let $(b, j_1, h_1, j_2, \dots, h_{\ell - 1}, j_\ell, i)$ be the shortest alternating path from $b$ to $i$
\State $\x_{h_{\ell - 1}} \gets \x_{h_{\ell - 1}} \setminus \{j_\ell\}$ \Comment{Chore transfer}
\State $\x_i \gets \x_i \cup \{j_\ell\}$
\State $b \gets \s{argmax}_{i\in N'} \min_{j\in \x_i} \p(\x_i\setminus j)$
\EndWhile  
\State $H_R \gets C_b \cap (N' \cup \x_{N'})$ \Comment{Partial component}
\State $N_R \gets H_R \cap N$ \Comment{Agent group}
\State $N' \gets N' \setminus N_R$, $R\gets R+1$
\EndWhile
\State \Return $(\x,\p, \{N_r\}_{r\in [R]})$
\end{algorithmic}
\end{algorithm}

Algorithm~\ref{alg:MakeInitGroups} starts with a cost-minimizing market outcome $(\x,\p)$ where each chore $j$ is assigned to an agent who has minimum cost for $j$. This ensures the allocation is fPO. The chore prices are set as follows. Each low-cost chore $j$ is assigned to an agent $i$ s.t. $c_{ij} = 1$. If an agent values all chores at $k$, then we can re-scale all values to $1$. Each low-cost chore is priced at 1, and each high-cost chore is priced at $k$. This pricing ensures that the mBB ratio of every agent is 1. The algorithm then eliminates pEF1-envy from the component of the big spender $b$ by identifying an agent $i$ in $C_b$ that is pEF1-envied by $b$, and transferring a single additional chore $j_\ell$ to $i$ from an agent $h_{\ell-1}$ who lies along a shortest alternating path from $b$ to $i$ (Lines 7 \& 8). Note that the identity of the big spender may change after transferring $j_\ell$ if $j_\ell$ belonged to $b$, so we must check who the big spender is after each transfer (Line 10). Once the component of the current big spender $b$ is pEF1, the same process is applied to the next big spender outside the previously made components. Repeated application of this process creates disjoint partial components $H_1, \ldots, H_R$ of agent sets $N_1,\dots,N_R$, where $R\le n$, all of which are pEF1. We refer to $N_1,\dots,N_R$ as agent \textit{groups}, and $H_1,\dots,H_R$ as initial (partial) components. Note also that the spending (up to the removal of the biggest chore) of the big spender $h_r$ of $H_r$ is weakly decreasing with $r$. We now record several properties of the output of Algorithm~\ref{alg:MakeInitGroups} .

\begin{restatable}{lemma}{initcompprop}\label{lem:initcompprop}
Algorithm~\ref{alg:MakeInitGroups} returns in $\poly{n,m}$-time a market outcome $(\x,\p)$ with agents partitioned into groups $N_1,\dots, N_R$, with the following properties:\normalfont
\begin{enumerate}
\item[(i)] For all low-cost chores $j\in M$, $p_j = 1$, and for all high-cost chores $j\in M$, $p_j = k$. 
\item[(ii)] The mBB ratio $\alpha_i$ of every agent $i$ is $1$.
\item[(iii)] Let $H_r$ be the collection of agents $N_r$ and chores allocated to them in $(\x,\p)$. Then each $H_r$ is a \textit{partial component} of some agent. That is, for each $r\in [R]$, there is an agent $h_r \in H_r$ s.t. $H_r$ comprises of all agents and chores not in $\bigcup_{r'<r} H_{r'}$ reachable through alternating paths from $h_r$. Further, $h_r$ is the big spender among agents not in $\bigcup_{r'<r} H_{r'}$:
\[h_r \in \s{argmax}_{i\notin (\bigcup_{r'<r} H_{r'})} \min_{j\in \x_i} \p(\x_i\setminus j)\]
\item[(iv)] The spending (up to removal of the largest chore) $f(r)$ of the big spender in $H_r$ weakly decreases with $r$. Here $f(r) = \max_{i\in H_r} \min_{j\in \x_i} \p(\x_i\setminus j)$.
\item[(v)] Each \textit{group is pEF1}, i.e., an agent does not pEF1-envy other agents in the same group.  
\item[(vi)] For every agent $i\in H_r$ and chore $j\in H_{r'}$ with $r'<r$, $c_{ij} = k$.
\item[(vii)] All high-cost chores belong to $H_R$.
\end{enumerate}
\end{restatable}
\begin{proof}
Properties (i) and (ii) follow from the construction of the initial allocation in Line 1.

Property (iii) follows from the construction of the set $H_r$ in Line 11. The agent $h_r$ is the last agent $b$ chosen in Line 10. Property (iv) then follows from the last result of (iii).

To see why (v) holds, we examine the loop in Lines 6-10. This loop terminates only once the big spender $b$ in the component does not pEF1-envy any other agent in the component. Since every agent spends (up to the removal of one chore) less than the big spender, it must be that every agent does not pEF1-envy any other agent in the component. Thus, the component is pEF1. 

Next, suppose for some agent $i\in H_r$ and some chore $j\in H_{r'}$ for $r'<r$, $c_{ij} = 1$. Then $j$ is a low-cost chore and must be priced at 1. Hence, $\alpha_{ij} = c_{ij}/p_j = 1 = \alpha_i$, implying that $j\in \mBB_i$. However this means that $i$ would have been added to $H_{r'}$, since there is an alternating path from $h_{r'}$ to $i$ via $j\in H_{r'}$. This is a contradiction, thus showing (vi).

Next, we show (vii). Suppose a high-cost chore $j$ belongs to $H_r$ for $r<R$. We know $j$ is priced at $k$. Then for an agent $i\in H_R$, $\alpha_{ij} = c_{ij}/p_j = 1 = \alpha_i$, implying that $j\in \mBB_i$. However this means $i$ belongs to $H_r$ since there is an alternating path from $h_r$ to $i$ via $j$, which is a contradiction.

Finally, we argue that in making a group starting with the a new big spender in $N'$ outside of previously established pEF1 groups, we do not disturb the previously established groups. Let $N_r$ be the current group being made, and let $N_{r'}$, where $r' < r$, be a previously established group. Let $b$ and $b'$ be the big spenders in $N_r$ and $N_{r'}$, respectively. We show that no agent in $N_{r'}$ gains or loses a chore. Note that there are no high-cost chores in $N_{r'}$, as high-cost chores are mBB for all agents so a group containing a high-cost chore must be the last group $N_R$. Since all chores are priced $1$ and $N_{r'}$ is pEF1, every agent in $N_{r'}$ has total spending either $\p(\x_{b'})$ or $\p(\x_{b'}) - 1$. We also have that $\p(\x_b) \leq \p(\x_{b'})$ by construction. Thus, it cannot be that $b$ pEF1-envies any agent in $N_{r'}$, so agents in $N_{r'}$ will not receive any chores. Agents in $N_{r'}$ also will not lose any chores, since they can only transfer chores to another agent in $N_{r'}$. However, we have already shown that agents in $N_{r'}$ cannot receive chores since they will not be pEF1-envied by $b$, so these transfers are impossible. Thus, all previous groups remain undisturbed while establishing a later group.
\end{proof}

Finally, we argue that:

\begin{restatable}{lemma}{inittermination}\label{lem:inittermination}
Algorithm~\ref{alg:MakeInitGroups} terminates in time $\poly{n,m}$.
\end{restatable}

\begin{algorithm}[ht]
\caption{Computing an EF1+fPO allocation}\label{alg:ef1po}
\textbf{Input:} Fair division instance $(N,M,C)$ with $c_{ij}\in \{1,k\}$\\
\textbf{Output:} An integral allocation $\x$
\begin{algorithmic}[1]
\State $(\x,\p, \{N_r\}_{r\in [R]}) \gets \s{MakeInitGroups(N,M,V)}$
\State $U \gets [R]$ \Comment{Unraised groups}
\State $\x^0 \gets \x$ \Comment{Copy of initial allocation, used in Line 22}
\State $b \gets \s{argmax}_{i\in N} \min_{j\in \x_i} \p(\x_i\setminus j)$ \Comment{Big Spender}
\State $\ell\gets \s{argmin}_{i\in N} \p(\x_i)$ \Comment{Least Spender}
\While{$(\x,\p)$ is not pEF1 \textbf{and} $\ell \in U$} 
\State $b \gets \s{argmax}_{i\in N} \min_{j\in \x_i} \p(\x_i\setminus j)$ 
\State $\ell\gets \s{argmin}_{i\in N} \p(\x_i)$
\State Let $(r,s)$ s.t. $b\in N_r$, $\ell \in N_s$ \Comment{$r < s$}
\If{$r \in U$}
\State Raise prices of chores owned by agents in $N_r$ by a factor of $k$ 
\State $U \gets U \setminus \{r\}$
\Else
\State Transfer a chore from $b$ to $\ell$ along an mBB edge 
\EndIf
\EndWhile
\While{$(\x,\p)$ is not pEF1}
\State $b \gets \s{argmax}_{i\in N} \min_{j\in \x_i} \p(\x_i\setminus j)$
\State $\ell\gets \s{argmin}_{i\in N} \p(\x_i)$
\State Let $(r,s)$ s.t. $b\in N_r$, $\ell \in N_s$
\If{$s > r$}
\State Transfer a chore from $b$ to $\ell$ along an mBB edge 
\ElsIf{$s < r$} 
\State $\exists \, i \in N_{r'}$ s.t. $r' \in U$ and $\exists \, j \in \x_i$ s.t. $j \in \x^0_{\ell}$ 
\State Transfer $j$ from $i$ to $\ell$
\State Transfer a chore from $b$ to $i$
\EndIf
\EndWhile
\State \Return $(\x,\p)$
\end{algorithmic}
\end{algorithm}

\subsection{Overview of Algorithm~\ref{alg:ef1po}}\label{sec:overview}
Our main algorithm (Algorithm~\ref{alg:ef1po}) begins by calling Algorithm~\ref{alg:MakeInitGroups}, which returns a market outcome $(\x,\p)$ and a set of agent groups $\{N_r\}_{r\in[R]}$ (with associated partial components $\{H_r\}_{r\in [R]}$) satisfying properties in Lemma~\ref{lem:initcompprop}. In the subsequent discussion, we refer to $(\x,\p)$ as the \textit{initial allocation}. Also in the subsequent discussion, all mentions of an agent receiving or losing chores are relative to this initial allocation. The following is an important invariant of Algorithm~\ref{alg:ef1po} (after the initial allocation is constructed).
\begin{restatable}{lemma}{spendingofLS}\label{lem:spendingofLS}
The spending of the least spender does not decrease in the run of Algorithm~\ref{alg:ef1po}.
\end{restatable}

We say that a group $N_r$ is above (resp. below) group $N_s$ if $r < s$ (resp. $r > s$). Lemma~\ref{lem:initcompprop} shows that each group $N_r$ is initially pEF1. Hence if the initial allocation $(\x,\p)$ is not pEF1, then the big spender $b$ and the least spender $\ell$ must be in different components. Since $b\in H_1$, it must be the case that $\ell\in H_s$ for some $s>1$. Since we want to obtain an fPO allocation, we can only transfer along mBB edges. Hence we raise the prices of all chores in $H_1$. We show that doing so creates an mBB edge from all agents $i\notin H_1$ to all chores $j\in H_1$ (Lemma~\ref{lem:priortot} below). In particular, there is an mBB edge from $\ell$ to a chore assigned to $b$. Hence we transfer a chore \textit{directly} from $b$ to $\ell$, thus reducing the pEF1-envy of $b$. This may change the identity of the big and least spenders. If the allocation is not yet pEF1, we must continue this process. 
 
At an arbitrary step in the run of the algorithm, let $b$ and $\ell$ be the big and least spenders. If the allocation is not pEF1, then $b$ pEF1-envies $\ell$ (Lemma~\ref{lem:BSenviesLS}). We consider cases based on the relative positions of $b$ and $\ell$. First we argue that $b$ and $\ell$ cannot lie in the same group, by showing that:

\begin{restatable}{lemma}{compsarepEFone}\label{lem:compsarepEF1}
Throughout the run of Algorithm~\ref{alg:ef1po}, each group $N_r$ remains pEF1.
\end{restatable}

Hence $b$ and $\ell$ must lie in different groups. Once again, since we want to transfer chores away from $b$ to reduce the pEF1-envy, and we want to obtain an fPO allocation, we only transfer chores along mBB edges. Doing so may require that we raise the prices of all chores belonging to certain agents in order to create new mBB edges to facilitate chore transfer. In our algorithm, \textit{all} agents in a group undergo price-rise together. We call a group $N_r$ a \textit{raised group} if its agents have undergone price-rise, else it is called an \textit{unraised group}. The set $U$ (Line 2) records the set of unraised components.

We will use the terms \textit{time-step} or \textit{iteration} interchangeably to denote either a chore transfer or a price-rise step. We say `at time-step $t$', to refer to the state of the algorithm \textit{just before} the event at $t$ happens. We denote by $(\x^{t},\p^t)$ the allocation and price vector at time-step $t$.

Let $T$ be the first time-step that the current LS enters a raised group. Note that such an event may or may not happen. Our algorithm performs instructions in Lines 6-14 before $T$, and Lines 15-24 after $T$, as we describe below.

\subsection{Algorithm prior to $T$ (Lines 6-14)}\label{sec:priortot}
We first record some properties of the algorithm prior to $T$. These observations follow directly from the algorithm.

\begin{restatable}{lemma}{algobs}\label{lem:algobs}
Prior to $T$, the following hold:
\begin{enumerate}
\item Any transfer of chores only takes place directly from the big spender $b$ to the least spender $\ell$. Thus, an agent receives a chore only if she is a least spender, and an agent loses a chore only if she is a big spender.
\item An agent ceases to be a least spender only if she receives a chore. An agent ceases to be a big spender only if she loses a chore.
\item A group undergoes price-rise at $t$ only if the group contains the big spender at $t$.
\end{enumerate}
\end{restatable}
We now show that:
\begin{restatable}{lemma}{BSincompofformerLS}\label{lem:BSincompofformerLS}
If at any point in the run of Algorithm~\ref{alg:ef1po} prior to $T$, the big spender lies in a group which contains a former least spender, then the allocation is pEF1.
\end{restatable}

This allows us to show that:
\begin{restatable}{lemma}{priortot}\label{lem:priortot}
Prior to $T$, the following hold:
\begin{enumerate}
\item[(i)] Let $r$ be the number of price-rise steps until time-step $t$, where $t<T$. Then the raised groups are exactly $N_1,\dots,N_r$. Furthermore they underwent price-rise exactly once and in that order.
\item[(ii)] For any chore $j$ allocated to an agent in a raised group $N_r$ and any agent $i$ in an unraised group $N_{r'}$, where $r'>r$, $j \in \mBB_i$.
\item[(iii)] For each $r'\in[r]$, at the time of price-rise of $N_{r'}$, no agent in $N_{r'}$ has either received or lost a chore since the initial allocation.
\end{enumerate}
\end{restatable}

\begin{proof}
We prove (i), (ii) and (iii) by induction. For $r=0$, they are trivially true since there are no raised groups. Assume that at some time-step $t$, (i) groups $N_1,\dots,N_{r}$ have undergone price-rise, once and in that order, for some $r\ge 1$, and (ii) and (iii) hold. 

Note that our algorithm only raises the prices of chores owned by a group if the group contains the big spender at the time (Lemma~\ref{lem:algobs}). If the current BS is in a raised group, then the induction hypothesis ensures that there is an mBB edge from the LS (who is in an unraised group prior to $T$) to chores owned by the BS. The algorithm therefore performs a direct chore transfers and no price-rise is necessary.

If eventually the BS enters an unraised group, then a price-rise step is potentially necessary. Suppose $b \in N_{r+1}$ is an agent who has received a chore prior to time-step $t$. If this happens then $b$ must have been a former LS by Lemma~\ref{lem:algobs}. Then Lemma~\ref{lem:BSincompofformerLS} shows that the allocation must already be pEF1. 

Hence we assume $b$ has not received a new chore since the initial allocation in Line 1. Furthermore since $b$ pEF1-envies the LS $\ell$, it must be the case that $\ell\in N_s$ where $s\ge r+2$. Lemma~\ref{lem:initcompprop} shows that there is no mBB edge from $\ell$ to chores owned by $b$, hence a direct chore transfer is not possible and it is necessary for $N_{r+1}$ to undergo price-rise. This shows (i).

Now if an agent $i\in N_{r+1}$ had previously received a chore, then $i$ is a former LS at $t$. Then $b\in N_{r+1}$ is in a group containing a former LS $i$, hence Lemma~\ref{lem:BSincompofformerLS} shows that the allocation is pEF1. Similarly Lemma~\ref{lem:algobs} shows that no agent in $N_{r+1}$ can have lost a chore. This is because only BS agents lose chores. Prior to $t$, no agent of $N_{r+1}$ can be the BS. Hence no agent in $N_{r+1}$ has received or lost a chore since the initial allocation at the time $N_{r+1}$ undergoes price-rise, thus showing (iii). 

The algorithm next raises the prices of all chores owned by $N_{r+1}$ by a factor of $k$, and $N_{r+1}$ becomes a raised group. Consider an agent $i\in N_{r'}$ for $r'\ge r+2$ and a chore allocated to an agent in $N_{r+1}$. Since the mBB ratio only changes upon a price-rise, the mBB ratio of $i$ is 1 since $N_{r'}$ does not undergo a price-rise before $N_{r+1}$. 

Observe that since $i\notin N_{r+1}$, there is no alternating path from agents in $N_{r+1}$ to $i$. Hence $j\notin \mBB_i$ before the price-rise. Thus $c_{ij}/p^t_j > 1$, showing $c_{ij} = k$ and $p^t_j = 1$. After the price-rise, we have that $p^{t+1}_j = k$, and $\alpha_i = c_{ij}/p^{t+1}_j$. Thus, $j\in \mBB_i$ after the price-rise, which shows (ii).
\end{proof}

To summarize the behavior of the Algorithm prior to $T$, we have argued in the above proof that if the allocation is not pEF1, we can always (i) transfer a chore directly from $b$ to $\ell$, or (ii) perform a price-rise on the group of $b$ and then transfer a chore from $b$ to $\ell$. Further, we argue that the algorithm makes progress towards getting a pEF1 allocation. 
\begin{lemma}\label{lem:priortoTtermination}
Algorithm~\ref{alg:ef1po} performs at most $\poly{n,m}$ steps prior to $T$.
\end{lemma}
\begin{proof}
Prior to $T$, the LS always remains in an unraised group. Chores are transferred away from agents who become big spenders in raised groups. Once an agent undergoes price-rise, she cannot gain any additional chores, since doing so would mean she is the LS in a raised group, which cannot happen prior to $T$. When the BS is in an unraised group, the group undergoes a price-rise. Thus, effectively, either agents in raised components only lose chores, or the BS `climbs-down' in the group list $N_1,\dots,N_R$, while the LS remains below the BS. Since there are $R\le n$ groups, and at most $m$ chores allocated to raised groups, after $\poly{n,m}$ steps either of two events happen: (i) the LS and BS both belong to the same group, or (ii) the LS enters a raised group. In the former case, the allocation is pEF1 due to Lemma~\ref{lem:compsarepEF1}, and the algorithm terminates in $\poly{n,m}$ steps. We discuss the latter case in the next section. Thus, there are at most $\poly{n,m}$ steps prior to $T$.
\end{proof}

\subsection{Algorithm after $T$ (Lines 15-24)}\label{sec:aftert}
We now describe the algorithm after $T$, i.e., once the LS enters a raised group (Lines 15-24). We show that subsequent to $T$, as long as the allocation is not pEF1, we can either (i) transfer a chore directly from $b$ to $\ell$, or (ii) transfer chores via an alternating path containing 3 agents. We do not perform any price-rises subsequent to $T$.

From Lemma~\ref{lem:priortot}, we know that at $T$, groups $N_1,\dots,N_r$ have undergone price-rise, for some $r\in[R]$. Let $N_{<r} = \bigcup_{r'< r} N_{r'}$, and $N_{>r} = \bigcup_{r'>r} N_{r'}$. The allocation at $T$ need not be pEF1, but we argue that it is already very close to being pEF1. Specifically, we show:

\begin{restatable}{lemma}{aboveHr}\label{lem:aboveHr}
At $T$, agents in $N_{<r}$ are pEF1 towards others.
\end{restatable}

\begin{restatable}{lemma}{belowHr}\label{lem:belowHr}
At $T$, agents in $N_{>r}$ are pEF1 towards others.
\end{restatable}

The above two lemmas imply that if the BS is not in $N_r$, then the allocation is pEF1. Let us assume that the allocation is not pEF1 at $T$. Let $b$, the BS at $T$ be in $N_r$, and $\ell$, the LS at $T$ be in $N_{<r}$, since the LS is in a raised group at $T$.

Suppose $\ell$ has never lost a chore. Let $\ell \in N_{r'}$, where $r'<r$, and $t'$ be the time when $N_{r'}$ underwent price-rise. Let $b'$ the BS at $t'$. Since the spending of the BS (up to removal of one chore) just after price-rises does not increase,
we have:
\[\p^T(\x^T_b\setminus j)\le \p^{t'}(\x^{t'}_{b'}\setminus j') \le \p^{t'}(\x^{t'}_\ell) = \p^T(\x^T_\ell), \]
for some chores $j\in \x^T_b, j'\in \x^{t'}_{b'}$. The intermediate transition follows from the property that $N_{r'}$ is pEF1. This shows that the allocation is pEF1.

On the other hand, suppose $\ell$ has lost at least one chore $j$ prior to $T$. At $T$, $j$ must be assigned to some unraised agent $i$ (Lemma~\ref{lem:priortot}). Further, there is a chore $j'\in \x^T_b$ s.t. $j'\in \mBB_i$. Thus, $b \xrightarrow{\x} j' \xrightarrow{\text{mBB}} i \xrightarrow{\x} j \xrightarrow{\text{mBB}} \ell$ is an alternating path. The algorithm now transfers chores along this alternating path. 

Note that as long as $\ell$ does not own a chore that she initially owned, such a path is available, and such a transfer is possible. If not, then it is as if $\ell$ has never lost a chore, and in that case the previous argument shows that the allocation must be pEF1.

Once again, doing such transfers may change the identities of the BS and LS, and we must continue the algorithm. We show that:
\begin{restatable}{lemma}{afterTinvariant}\label{lem:afterTinvariant}
After $T$, the following are invariant:
\begin{enumerate}
\item[(i)] Agents in $N_{<r}$ do not pEF1-envy any other agent.
\item[(ii)] Agents in $N_{>r}$ do not pEF1-envy any other agent.
\item[(iii)] Each group is pEF1.
\end{enumerate}
\end{restatable}

Just after $T$, the BS is in $N_r$ and the LS is in $N_{<r}$. After a chore transfer, the identity of the LS or BS can change. If the BS enters either $N_{<r}$ or $N_{>r}$, then using Lemma~\ref{lem:afterTinvariant} the allocation would be pEF1. While the BS is in $N_r$: (i) if the LS is in $N_r$, the allocation would be pEF1, (ii) if the LS is in $N_{>r}$, then we can transfer from BS to LS directly along an mBB edge (which exists due to Lemma~\ref{lem:priortot}), (iii) if the LS is in $N_{<r}$, then we can transfer from the BS to LS via an alternating path with three agents as described above.

Finally we argue termination in polynomial-time:

\begin{lemma}\label{lem:afterTtermination}
Algorithm~\ref{alg:ef1po} performs at most $\poly{n,m}$ steps after $T$ and terminates with a pEF1 allocation.
\end{lemma}
\begin{proof}
Call the difference between the spending (up to the removal of the biggest chore) of the big spender and the spending of the least spender as the \textit{spending gap}. If the allocation is not pEF1, the spending gap is positive. After $T$, there are no price-rises, hence the spending gap weakly decreases. As long as the allocation is not pEF1, the BS must be in $N_r$. Based on whether the LS is in $N_{<r}$ or $N_{>r}$ we can perform chore transfers which weakly decrease the spending gap. While the allocation is not pEF1, such a transfer is always possible. Further, each such transfer reduces the number of chores owned by agents in $N_r$, and such agents do not receive any chores again. Hence there can only be $\poly{n,m}$ steps after $T$ eventually terminating in a pEF1 allocation.
\end{proof}

\subsection{Summarizing our EF1+fPO algorithm}
We summarize Algorithm~\ref{alg:ef1po}.

\begin{enumerate}
\item Algorithm~\ref{alg:ef1po} first calls Algorithm~\ref{alg:MakeInitGroups} to partition agents into groups $N_1,\dots,N_R$ with properties as in Lemma~\ref{lem:initcompprop} to obtain an initial allocation. Lemma~\ref{lem:inittermination} shows this takes $\poly{n,m}$ steps.
\item When the current allocation $(\x,\p)$ is not pEF1, the BS $b$ pEF1-envies the LS $\ell$. While there is an mBB edge from $\ell$ to a chore owned by $b$, we transfer a chore directly from $b$ to $\ell$. If not, in order to transfer along mBB edges, we may have to raise the prices of chores belonging to the group of $b$, creating \textit{raised} groups.
\item Let $T$ be the first time-step when the LS enters a raised group. Prior to $T$, while the allocation is not pEF1, the algorithm either performs a direct chore transfer from the BS to the LS, or performs price-rise on the group of $b$. Lemma~\ref{lem:priortot} shows that the groups are raised exactly once and in order of $N_1,\dots,N_R$. Lemma~\ref{lem:priortoTtermination} shows the algorithm runs for $\poly{n,m}$ steps before $T$.
\item Once the LS enters a raised group at $T$, there are no more price-rise steps. The algorithm then performs chore transfers from the BS to LS via alternating paths with at most 3 agents. Lemma~\ref{lem:afterTinvariant} and Lemma~\ref{lem:afterTtermination} show that the algorithm performs at most $\poly{n,m}$ steps after $T$ and terminates with a pEF1 allocation.
\item Finally, we note that allocation is always fPO, since (i) Algorithm~\ref{alg:MakeInitGroups} returns a market outcome which is fPO, and (ii) any transfer of chores happens along mBB edges. 
\end{enumerate}

\subsection{Examples}
We illustrate two examples of an execution of Algorithm~\ref{alg:ef1po} for a fair division instance, one in which the algorithm terminates before a time $T$ where the LS enters a raised component, and one in which the algorithm terminates after such a time $T$.

\subsubsection{Algorithm terminates before $\mathbf{T}$}
Consider the instance captured by the table of disutilities below.
\begin{center}
\begin{tabular}{|c|c|c|c|c|c|c|c|c|c|c|c|c|c|c|}
\hline
\multicolumn{2}{|c|}{} & \multicolumn{13}{|c|}{Chores} \\
\cline{3-15}
 \multicolumn{2}{|c|}{} & $j_1$ & $j_2$ & $j_3$ & $j_4$ & $j_5$ & $j_6$ & $j_7$ & $j_8$ & $j_9$ & $j_{10}$ & $j_{11}$ & $j_{12}$ & $j_{13}$ \\ 
\hline
\multirow{6}{3em}{Agents} & $a_1$ & 1 & 1 & 1 & 1 & 1 & $k$ & $k$ & $k$ & $k$ & $k$ & $k$ & $k$ & $k$ \\
\cline{2-15}
 & $a_2$ & $k$ & $k$ & $k$ & $k$ & $k$ & 1 & 1 & 1 & 1 & $k$ & $k$ & $k$ & $k$ \\
\cline{2-15}
 & $a_3$ & $k$ & $k$ & $k$ & $k$ & $k$ & $k$ & $k$ & $k$ & $k$ & 1 & $k$ & $k$ & $k$ \\
\cline{2-15}
 & $a_4$ & $k$ & $k$ & $k$ & $k$ & $k$ & $k$ & $k$ & $k$ & $k$ & $k$ & 1 & $k$ & $k$ \\
\cline{2-15}
 & $a_5$ & $k$ & $k$ & $k$ & $k$ & $k$ & $k$ & $k$ & $k$ & $k$ & $k$ & $k$ & 1 & $k$ \\
\cline{2-15}
 & $a_6$ & $k$ & $k$ & $k$ & $k$ & $k$ & $k$ & $k$ & $k$ & $k$ & $k$ & $k$ & $k$ & 1 \\
\hline
\end{tabular}
\end{center}
\hspace{1cm}
Algorithm~\ref{alg:ef1po} begins by calling Algorithm~\ref{alg:MakeInitGroups} which returns the following initial allocation $\x$.
\begin{multicols}{2}
\begin{itemize}
    \item $\x_{a_1} = \{j_1, j_2, j_3, j_4, j_5\}$
	\item $\x_{a_2} = \{j_6, j_7, j_8, j_9\}$
	\item $\x_{a_3} = \{j_{10}\}$
	\item $\x_{a_4} = \{j_{11}\}$
	\item $\x_{a_5} = \{j_{12}\}$
	\item $\x_{a_6} = \{j_{13}\}$
\end{itemize}
\end{multicols}

Suppose $k = 5$. Note that here each agent lies in its own group. Since $\x$ is not pEF1, Algorithm~\ref{alg:ef1po} proceeds to raises the prices of the chores of $a_1$ by a factor of $k$. Then, a chore is transferred from $a_1$ to each of $a_3$, $a_4$, $a_5$, and $a_6$. This gives us the following pEF1 (and thus EF1) allocation $\x'$.
\begin{multicols}{2}
\begin{itemize}
    \item $\x'_{a_1} = \{j_5\}$
	\item $\x'_{a_2} = \{j_6, j_7, j_8, j_9\}$
	\item $\x'_{a_3} = \{j_1, j_{10}\}$
	\item $\x'_{a_4} = \{j_2, j_{11}\}$
	\item $\x'_{a_5} = \{j_3, j_{12}\}$
	\item $\x'_{a_6} = \{j_4, j_{13}\}$
\end{itemize}
\end{multicols}

\subsubsection{Algorithm terminates after $\mathbf{T}$}
We now consider a slight alteration of the previous instance, with an additional agent $a_7$ and additional chore $j_{14}$. The table of disutilities is given below.

\begin{center}
\begin{tabular}{|c|c|c|c|c|c|c|c|c|c|c|c|c|c|c|c|}
\hline
\multicolumn{2}{|c|}{} & \multicolumn{14}{|c|}{Chores} \\
\cline{3-16}
 \multicolumn{2}{|c|}{} & $j_1$ & $j_2$ & $j_3$ & $j_4$ & $j_5$ & $j_6$ & $j_7$ & $j_8$ & $j_9$ & $j_{10}$ & $j_{11}$ & $j_{12}$ & $j_{13}$ & $j_{14}$ \\ 
\hline
\multirow{7}{3em}{Agents} & $a_1$ & 1 & 1 & 1 & 1 & 1 & $k$ & $k$ & $k$ & $k$ & $k$ & $k$ & $k$ & $k$ & $k$ \\
\cline{2-16}
 & $a_2$ & $k$ & $k$ & $k$ & $k$ & $k$ & 1 & 1 & 1 & 1 & $k$ & $k$ & $k$ & $k$ & $k$ \\
\cline{2-16}
 & $a_3$ & $k$ & $k$ & $k$ & $k$ & $k$ & $k$ & $k$ & $k$ & $k$ & 1 & $k$ & $k$ & $k$ & $k$ \\
\cline{2-16}
 & $a_4$ & $k$ & $k$ & $k$ & $k$ & $k$ & $k$ & $k$ & $k$ & $k$ & $k$ & 1 & $k$ & $k$ & $k$ \\
\cline{2-16}
 & $a_5$ & $k$ & $k$ & $k$ & $k$ & $k$ & $k$ & $k$ & $k$ & $k$ & $k$ & $k$ & 1 & $k$ & $k$ \\
\cline{2-16}
 & $a_6$ & $k$ & $k$ & $k$ & $k$ & $k$ & $k$ & $k$ & $k$ & $k$ & $k$ & $k$ & $k$ & 1 & $k$ \\
\cline{2-16}
 & $a_7$ & $k$ & $k$ & $k$ & $k$ & $k$ & $k$ & $k$ & $k$ & $k$ & $k$ & $k$ & $k$ & $k$ & 1 \\
\hline
\end{tabular}
\end{center}
\hspace{1cm}

Again, Algorithm~\ref{alg:ef1po} first calls Algorithm~\ref{alg:MakeInitGroups} which returns the following initial allocation $\y$.

\begin{multicols}{2}
\begin{itemize}
    \item $\y_{a_1} = \{j_1, j_2, j_3, j_4, j_5\}$
	\item $\y_{a_2} = \{j_6, j_7, j_8, j_9\}$
	\item $\y_{a_3} = \{j_{10}\}$
	\item $\y_{a_4} = \{j_{11}\}$
	\item $\y_{a_5} = \{j_{12}\}$
	\item $\y_{a_6} = \{j_{13}\}$
	\item $\y_{a_7} = \{j_{14}\}$
	\item[~] ~
\end{itemize}
\end{multicols}

Again let $k=5$. As before, Algorithm~\ref{alg:ef1po} raises the prices of the chores of $a_1$ by a factor of $k$ and transfers a chore from $a_1$ to each of $a_3$, $a_4$, $a_5$, and $a_6$. Now, however, the allocation is not pEF1. The BS, $a_2$, still pEF1-envies the LS $a_7$. Thus, the prices of $a_2$'s chores are raised, and a chore is transferred from $a_2$ to $a_7$. We now have the following allocation $\y'$. 

\begin{multicols}{2}
\begin{itemize}
    \item $\y'_{a_1} = \{j_5\}$
	\item $\y'_{a_2} = \{j_7, j_8, j_9\}$
	\item $\y'_{a_3} = \{j_1, j_{10}\}$
	\item $\y'_{a_4} = \{j_2, j_{11}\}$
	\item $\y'_{a_5} = \{j_3, j_{12}\}$
	\item $\y'_{a_6} = \{j_4, j_{13}\}$
	\item $\y'_{a_7} = \{j_6, j_{14}\}$
	\item[~] ~
\end{itemize}
\end{multicols}

We see that $\y'$ is not pEF1 as $a_2$ pEF1-envies $a_1$, and in fact $a_1$ is the LS. This marks the first time $T$ in which the LS has entered a raised group. Then, to achieve pEF1, $j_1$ is returned to $a_1$ from $a_3$, and $j_7$ is transferred from $a_2$ to $a_3$ as a replacement. This gives the pEF1 allocation $\y''$.
\begin{multicols}{2}
\begin{itemize}
    \item $\y''_{a_1} = \{j_1, j_5\}$
	\item $\y''_{a_2} = \{j_8, j_9\}$
	\item $\y''_{a_3} = \{j_7, j_{10}\}$
	\item $\y''_{a_4} = \{j_2, j_{11}\}$
	\item $\y''_{a_5} = \{j_3, j_{12}\}$
	\item $\y''_{a_6} = \{j_4, j_{13}\}$
	\item $\y''_{a_7} = \{j_6, j_{14}\}$
	\item[~] ~
\end{itemize}
\end{multicols}

\section{EF+PO allocation of divisible chores}\label{sec:divisible}
In this section, we prove our second result:
\begin{restatable}{theorem}{divisible}\label{thm:divisible}
Given a bivalued fair division instance $(N,M,V)$ of divisible chores with all $c_{ij} \in \{a,b\}$ for some $a,b\in \Rp$, an EF+fPO allocation can be computed in strongly polynomial-time.
\end{restatable}

We prove Theorem~\ref{thm:divisible} by showing that Algorithm~\ref{alg:efpo} computes an EF+fPO allocation in strongly polynomial-time. First note that we can scale the costs of the chores so that they are in $\{1, k\}$.

\subsection{Describing the Approach}\label{sec:divisible-approach}
Since we are interested in an EF+fPO allocation, we try to construct a (fractional) market outcome $(\x,\p)$ which is price-envy-free (pEF) as opposed to pEF1. A market outcome $(\x,\p)$ is said to be pEF if for all agents $i,h\in N$, $\p(\x_i)\ge \p(\x_h)$, i.e., all agents have equal spending. Similar to Lemma~\ref{lem:pEF1impliesEF1}, we can show that if $(\x,\p)$ is a pEF market outcome, then $\x$ is EF. Thus, our aim is to obtain a market outcome in which all agents have the same spending. 

As in Algorithm~\ref{alg:ef1po}, we begin with an initial market outcome $(\x,\p)$ on mBB, with certain desirable properties. Specifically, we partition the agents into groups $N_1,\dots,N_R$ so that each group spends the same amount, each group is a partial component of some agent, and the \textit{per-agent spending} of the group $N_r$ decreases with $r$. Now, if the allocation is not pEF, it must be the case that any agent in the \textit{pool} of biggest spenders $B$ (here $B = \{b : b\in \s{argmax}_{i\in N} \p(\x_i)\}$) envies any agent in the pool of least spenders $L$ (here $L = \{ \ell : \ell\in \s{argmin}_{i\in N} \p(\x_i) \}$), i.e., $\p(\x_b) > \p(\x_\ell)$ for any $b\in B$ and $\ell \in L$. The natural idea following Algorithm~\ref{alg:ef1po} is to reduce the pEF-envy by \textit{draining} chores uniformly from $B$ to $L$ along mBB edges. Further, if the pool $B$ has chores which are not on mBB for $L$, then we raise the prices of certain chores in $B$ by a factor of $k$, which creates the required mBB edges following Lemma~\ref{lem:priortot}. Initially, $B$ equals $N_1$, and $L$ equals $N_R$. At a certain point in the run of the algorithm, $B$ equals $\bigcup_{i\le r} N_i$ and $L$ equals $\bigcup_{i\ge r'} N_{r'}$, where $r\le r'$. When eventually the pools of biggest and least spenders coincide, the algorithm terminates with a pEF allocation. This approach is similar to Algorithm~\ref{alg:ef1po} and terminates in polynomial time.

However, we observe a stronger property of this approach, which we use to design Algorithm~\ref{alg:efpo}. We claim that effectively, this approach is equivalent to raising the prices of all chores in the first $r^*$ groups, and then draining chores carefully, eventually resulting in a pEF allocation. Therefore, we can simply iterate over $r=1$ to $R-1 \le n$ as a guess for a right value of $r^*$. We now describe the parts of the algorithm and also prove the above argument.

\subsection{Obtaining Initial Groups}\label{sec:divisible-initcomps}
\begin{algorithm}[t]
\caption{$\s{MakeInitGroupsDiv}$}\label{alg:MakeInitGroupsDiv}
\textbf{Input:} Fair division instance $(N,M,C)$ with $c_{ij}\in \{1,k\}$\\
\textbf{Output:} Fractional alloc. $\x$, prices $\p$, agent groups $\{N_r\}_{r\in [R]}$
\begin{algorithmic}[1]
\State $(\x,\p) \gets$ initial cost minimizing balanced flow fractional market allocation, where $p_j = c_{ij}$ for $j \in \x_i$.
\State $R\gets 1$, $N' \gets N$
\While{$N' \neq \emptyset$}
\State $b \gets \s{argmax}_{i \in N'} \p(\x_i)$ \Comment{Biggest Spender}
\State $C_b \gets$ Component of $b$ \Comment{See Definition~\ref{def:component}}
\State $H_R \gets C_b \cap (N' \cup \x_{N'})$ \Comment{Partial component}
\State $N_R \gets H_R \cap N$ \Comment{Agent group}
\State $N' \gets N' \setminus N_R$, $R\gets R+1$
\EndWhile
\State \Return $(\x,\p, \{N_r\}_{r\in [R]})$
\end{algorithmic}
\end{algorithm}

\begin{algorithm}[t!]
\caption{Computing an EF+fPO allocation}\label{alg:efpo}
\textbf{Input:} Fair division instance $(N,M,C)$ with $c_{ij}\in \{1,k\}$\\
\textbf{Output:} A fractional allocation $\x$
\begin{algorithmic}[1]
\For {$r \in [R-1]$}
	\State $(\x,\p, \{N_r\}_{r\in [R]}) \gets \s{MakeInitGroupsDiv(N,M,V)}$
	\State Raise prices of chores of agents in $\{N_i\}_{i \in [r]}$ by a factor of $k$
	\State $B \gets N_1$ \Comment{Pool of Biggest Spenders}
	\State $\beta \gets 1$ \Comment{Index of lowest group in $B$}
	\State $L \gets N_R$ \Comment{Pool of Least Spenders}
	\State $\lambda \gets R$ \Comment{Index of highest group of $L$}
	\While{$\beta \leq r$ \textbf{and} $\lambda \geq r + 1$}
		\State $d_B \gets |B| \cdot (s(B) - s(N_{\beta+1}))$ \State $d_L \gets |L| \cdot (s(N_{\lambda-1})) - s(L))$ 
		\For{agents $b \in B$}
			\If{$\beta = r$ \textbf{and} $\lambda = r + 1$}
				\State $q \gets \frac{|B| \cdot s(B) + |L| \cdot s(L)}{|N|}$
				\State $C \gets$ subset of $\x_{b}$ s.t. $\p(C) = s(B) - q$
			\ElsIf{$d_B \geq d_L$}
				\State $C \gets$ subset of $\x_{b}$ s.t. $\p(C) = \frac{d_L}{|B|}$
			\Else
				\State $C \gets$ subset of $\x_{b}$ s.t. $\p(C) = \frac{d_B}{|B|}$
			\EndIf
			\For{agents $\ell \in L$}
				\State $\x_{b} \gets \x_{b} \setminus \frac{1}{|L|}C$
				\State $\x_{\ell} \gets \x_{\ell} \cup \frac{1}{|L|}C$
			\EndFor
		\EndFor
		\If{$(\x, \p)$ is pEF}
			\State \Return $(\x, \p)$		
		\ElsIf{$d_B \geq d_L$}
			\State $L \gets L \cup N_{\lambda-1}$
			\State $\lambda \gets \lambda - 1$
		\Else
			\State $B \gets B \cup N_{\beta+1}$
			\State $\beta \gets \beta + 1$
			
		\EndIf	
	\EndWhile
\EndFor
\end{algorithmic}
\end{algorithm}

Algorithm~\ref{alg:MakeInitGroupsDiv} obtains a partition of the set $N$ of agents into groups $N_1, \ldots, N_R$ similar to Algorithm~\ref{alg:MakeInitGroups}, with the strengthened condition that the partial components are price envy-free (pEF) rather than pEF1. The price vector $\p$ is set the same way as in Algorithm~\ref{alg:MakeInitGroups}. Then all low-cost chores are priced at $1$ and high-cost chores are priced at $k$. To allocate these chores, however, we use a balanced flow formulation~\cite{devanur2008mkteqb}. This gives us an allocation with the key property that if agent $i$ spends more than agent $j$, then there is no alternating path from $i$ to $j$. Then, each $b$ selected in Line 4 of Algorithm~\ref{alg:MakeInitGroupsDiv} is a biggest spender in its partial component, but also cannot spend more than any agent in its partial component. Thus, it must be that the partial component (i.e., group) is pEF. Algorithm~\ref{alg:MakeInitGroupsDiv} runs in polynomial time as the balanced flow allocation is achieved by at most $m$ max-flow computations (see~\cite{devanur2008mkteqb}) and the remaining structure is analogous to Algorithm~\ref{alg:MakeInitGroups}.

\subsection{Discussion on Algorithm~\ref{alg:efpo}}\label{sec:divisible-discussion}
Algorithm~\ref{alg:efpo} proceeds similarly to Algorithm~\ref{alg:ef1po} with two key differences:
\begin{enumerate}
	\item[(1)] Multiple groups are raised simultaneously rather than one at a time.
	\item[(2)] Transfer are made between a \textit{pool of biggest spenders} and a \textit{pool of least spenders}. That is, we transfer (fractions of) chores from \textit{all} biggest spenders to \textit{all} least spenders. 
\end{enumerate}
We show that there always exists some $r^*$, $1 \leq r^* \leq R - 1$, such that simultaneously raising the first $r^*$ groups allows us to reach a pEF allocation by performing chore transfers along mBB edges from biggest spenders to least spenders. Recalling that the partial components associated with each group $N_i$ are pEF, let $s(N_i)$ denote the \textit{per agent spending} of $N_i$. By construction, we have that $s(N_i) \geq s(N_j)$ for $i \leq j$. Suppose we do an initial raise of the first $r$ groups. Let $\beta \in [n]$ be such that the pool of biggest spenders $B$ is the union of raised groups $\bigcup_{i=1}^\beta N_i$. Similarly, let $\lambda \in [n]$ be such that the pool of least spenders $L$ is the union of unraised groups $\bigcup_{i=\lambda}^R N_i$. Algorithm~\ref{alg:efpo} then proceeds to transfer chores out of agents in $B$ at a uniform rate $\rho_B$ and into agents in $L$ at a uniform rate $\rho_L$, such that $|B|\rho_B = |L|\rho_L$. Then, $s(B)$ will fall towards $s(N_{\beta+1})$ while $s(L)$ rises towards $s(N_{\lambda - 1})$. When we transfer enough chores so that $s(B)$ reaches $s(N_{\beta+1})$ or $s(L)$ reaches $s(N_{\lambda-1})$, whichever occurs first, we add the agents in that group to the respective pool. Note that in this process both $B$ and $L$ are only growing. Once an agent enters $B$ or $L$, she never leaves. In addition, we have that $s(B)$ is always decreasing while $s(L)$ is always increasing, i.e., the \textit{spending gap} $s(B) - s(L)$ is weakly decreasing. When the spending gap is zero, i.e. $s(B) = s(L)$, the allocation is pEF. However, to guarantee that transfers occur along mBB edges, we must have that all agents in $B$ are in raised groups, and all agents in $L$ are in unraised groups. Thus, we give the two conditions under which the algorithm \textit{cannot} obtain a pEF allocation after raising the first $r$ groups:
\begin{enumerate}
    \item[-] \textit{Condition 1.} Raised groups have \textit{too much} total spending. The spending of agents in $L$ rises to spending level of $N_{r}$ before the spending of agents in $B$ falls to the level of $N_{r}$. That is, a raised group enters $L$.
    \item[-] \textit{Condition 2.} Raised groups have \textit{too little} total spending. The spending of agents in $B$ falls to the level of $N_{r+1}$ before the spending of agents in $L$ rises to the level of $N_{r+1}$. That is, an unraised group enters $B$.
\end{enumerate}
If neither of these conditions hold, then $B$ always contains only raised groups, and $L$ always holds only unraised groups. It is clear then that we can perform transfers until the spending gap between $B$ and $L$ closes entirely and we have a pEF allocation. Specifically, this is possible when the algorithm reaches a point where $B$ contains all the raised groups and $L$ contains all the unraised groups (see Line 12). We show that is guaranteed for some $r^*$, $1 \leq r^* \leq R-1$.

Suppose we raise only the single group $N_1$. It cannot be that Condition 1 holds here, as $L$ rising to the level of $N_1$ implies that the allocation is pEF. It must be that we have Condition 2, or we would be done as the allocation will be pEF and the algorithm would halt.

Now suppose inductively that we have raised the first $r$ groups and found that the allocation is not pEF and Condition 1 is impossible for the first $r$ groups, i.e., Condition 2 holds. Then, it must also be that Condition 1 is impossible if we raise the first $r+1$ groups if the allocation is still not pEF. Let $\sigma = s(N_{r+1})$ be the per-agent spending of $N_{r+1}$ just before $N_{r+1}$ undergoes price-rise. Since raising $r$ groups results in Condition 2, we know that $s(B)$ falls to $\sigma$ \textit{before} $s(L)$ rises to $\sigma$. Surely then, in raising $r+1$ groups, $s(B)$ falls to $\sigma \cdot k$ before $s(L)$ rises to $\sigma \cdot k$, as $\sigma \cdot k > \sigma$. This shows Condition 1 cannot hold after raising the first $r+1$ groups, given that the allocation is not yet pEF.

Finally, suppose that we have raised the first $R-1$ groups, i.e., all but the last group, and the allocation is not pEF. Then, Condition 2 cannot hold, as $s(B)$ falling to the level of $s(N_R)$ implies that the allocation is pEF. Thus, if the allocation is not pEF, Condition 1 must hold, but this is impossible by our inductive argument. Hence it must be that for some $r^*$, $1 \leq r^* \leq R-1$, a pEF allocation was obtained. 

Following this logic, Algorithm~\ref{alg:efpo} checks each possible $r$ until $r^*$ is found and a pEF allocation is obtained. Further, since all chore transfers are along mBB, the allocation remains fPO throughout. Thus the algorithm computes a pEF+fPO allocation. 

For each $r \in [R]$, the algorithm performs at most $O(n^2)$ transfers before adding an agent group to either $B$ or $L$. A group cannot be added to both $B$ and $L$, nor can it leave $B$ or $L$ after being added. Thus, since $R \leq n$, the total number of transfers performed to check if some $r$ is correct is $O(n^3)$, so checking all possible $r$ takes $O(n^4)$ transfers. Hence, Algorithm~\ref{alg:efpo} terminates in strongly polynomial-time, proving Theorem~\ref{thm:divisible}.

\subsection{Pseudocode Description}
We describe the details of the pseudocode of Algorithm~\ref{alg:efpo}. 

As described in Section~\ref{sec:divisible-approach}, Algorithm~\ref{alg:efpo} iteratively searches for a right value $r^*$ of the number of groups to raise. This corresponds to the for loop starting in Line 1. Line 2 calls Algorithm~\ref{alg:MakeInitGroupsDiv} to obtain initial groups as described in Section~\ref{sec:divisible-initcomps}. 

Lines 4-5 initialize the pool of biggest spenders $B$ to $N_1$ and set the index of the lowest group in $B$, $\beta$, to 1. Lines 6-7 initialize the pool of least spenders $L$ to $N_R$ and set the index of the lowest group in $L$, $\lambda$, to $R$.

The loop starting in Line 8 is active while the pool of biggest spenders doesn't contain an unraised component ($\beta \le r$) and the pool of least spenders doesn't contain a raised component ($\lambda \ge r+1$). We then want to decide how much spending (chores) to \textit{drain} from $B$ to $L$. As described in Section~\ref{sec:divisible-discussion}, either the per-agent spending of $B$ can climb down to the per-agent spending of $N_{\beta+1}$, or the per-agent spending of $L$ can climb up to the per-agent spending of $N_{\lambda-1}$. Line 9 computes the total amount of spending $d_B$ to be transferred in the former case and Line 10 computes the total amount of spending $d_L$ to be transferred in the latter case.

The loop starting in Line 11 considers agents $b\in B$ one by one and computes an appropriately priced bundle of chores that should be taken away from them. If $B$ contains all the raised groups and $L$ contains all the unraised groups, i.e., if $\beta = r$ and $\lambda = r+1$ (Line 12), then we are a step away from obtaining a pEF allocation. We compute the average spending of all agents, $q$, in Line 13. This is the amount to which we want to bring every agent's spending. Hence, for each agent $b$ in consideration, we compute a subset $C\subseteq \x_b$ with value $s(B)-q$ (Line 14), and distribute it evenly to all agents in $L$ (Lines 20-21). Here, $\frac{1}{|L|}C$ denotes a (fractional) bundle of chores $C'$ where $C'_{ij} = \frac{1}{|L|}C_{ij}$ for agent $i$ and chore $j$ (revisit Section~\ref{sec:prelim} for the definition of a fractional allocation).

Otherwise, if $d_B \ge d_L$, then the per-agent spending of $L$ reaches the per-agent spending of $N_{\lambda-1}$ before the per-agent spending of $B$ reaches the per-agent spending of $N_{\beta+1}$. Hence each agent $\ell \in L$ must obtain an additional spending of $\frac{d_L}{|L|}$ from agents in $B$. We do this by computing a subset $C$ of $\x_b$ of value $\frac{d_L}{|B|}$ from every $b\in B$ (Line 16) and giving a $\frac{1}{|L|}$-sized fraction of $C$ to $\ell$ (Line 20-21). Thus, $\ell$ obtains a spending $\frac{1}{|L|}\cdot\frac{d_L}{|B|}$ from every agent in $B$, which totals to $\frac{d_L}{|L|}$, the required amount. Now, we update the pool of least spenders to include $N_{\lambda-1}$, and update $\lambda$ (Lines 25-26).

Analogously, when $d_B < d_L$, the per-agent spending of $B$ reaches the per-agent spending of $N_{\beta+1}$ before the per-agent spending of $L$ reaches the per-agent spending of $N_{\lambda-1}$. Hence, each agent $b\in B$ loses an amount equalling $\frac{d_B}{|B|}$. For each $b$, we compute a set $C$ with spending $\frac{d_B}{|B|}$ (Line 18), which is then transferred equally among all agents in $L$ (Line 20-21). We also update the pool of biggest spenders to include $N_{\beta+1}$, and update $\beta$ (Lines 28-29). 

We terminate if the allocation is pEF (Lines 22-23). If not, then we continue our search for a right value $r^*$ by incrementing $r$ (Line 1). The arguments of Section~\ref{sec:divisible-discussion} show that for some $1\le r \le R-1$, Algorithm~\ref{alg:efpo} will terminate with a pEF+fPO allocation.

\section{Discussion}
In this paper, we presented a strongly polynomial-time algorithm for computing an EF1+fPO allocation of indivisible chores to agents with bivalued preferences, constituting the first non-trivial result for the EF1+PO problem for chores. Our algorithm is novel and relies on several involved arguments. Given that the general case is a challenging open problem, we believe extending our algorithm and its analysis to the class of $k$-ary chores is an interesting and natural next step. Another interesting question is whether we can compute an EFX allocation in this setting. We also presented a strongly polynomial-time for computing an EF+fPO allocation of divisible bivalued chores. Computing an EF+fPO allocation of divisible $k$-ary chores in polynomial-time is also a compelling direction for future work.

\appendix
\setcounter{secnumdepth}{2}
\renewcommand{\thesection}{\Alph{section}}
\setcounter{section}{0}

\section{Missing Proofs from Section~\ref{sec:indivisible}}\label{app:indivisible}
\inittermination*
\begin{proof} 
Consider the time needed to make the component of the big spender pEF1. We show that this takes time $\poly{n, m}$. Note that when the identity of the big spender $b$ does not change, a transfer of chore $j$ strictly increases the level of $j$. Thus, while the big spender $b$ is unchanging, chore $j$ can be transferred at most $n$ times. Thus, the total number of chore transfers with the same big spender is at most $mn$. 

We now bound the total number of times the identity of the big spender may change. Let $\x^t$ be the allocation at iteration $t$, and $\p$ be the price vector. Furthermore, let $j_{(i, t)}$ denote the highest-priced chore of agent $i$ at time $t$. 

We claim that if an agent $b$ ceases to be the BS at iteration $t$ and subsequently becomes the BS again at $t'$, then her spending (up to removal of one chore) strictly goes down, i.e., $\p(\x^{t'}_b \setminus j_{(b, t')}) < \p(\x^{t}_b \setminus j_{(b, t)})$. Note that $b$ must lose some chore at time $t$ to cease being the BS. Then, since the spending of $b$ up to one chore at time $t'$ is weakly greater than at time $t$, $b$ must also have gained a chore between $t$ and $t'$. Let the most recent time $b$ gained a chore be $t''\in (t,t')$, and let $b'$ (resp. $b''$) be the BS at time $t'$ (resp. $t''$). Note that the spending of the big spender up to one chore weakly decreases with each iteration, so $\p(\x^{t'}_{b'} \setminus j_{(b', t')}) \leq \p(\x^{t''}_{b''} \setminus j_{(b'', t'')}) \leq \p(\x^{t}_{b} \setminus j_{(b, t)})$. Since $b$ receives a chore $j$ at time $t''$, $b$ must be envied by $b''$ at time $t''$ (Line 6). This gives us $\p(x^{t''}_b) < \p(\x^{t''}_{b''} \setminus j_{(b'', t'')})$. As $j$ is the last chore gained by $b$, we have $\p(\x^{t'}_{b} \setminus j_{(b, t')}) \leq \p(\x^{t'}_b \setminus j) \leq \p(\x^{t''}_b) < \p(\x^{t''}_{b''} \setminus j_{(b'', t'')}) \leq \p(\x^{t}_b \setminus j_{(b, t)})$, thus proving the claim.

Thus, if an agent $i$ ceases to be the BS at time $t$ and later becomes the BS again at $t'$, the spending of $i$ is strictly smaller at $t'$ than at $t$. It follows that $i$ has strictly smaller disutility at $t'$ since the prices have not changed. In any allocation $\x$, if $s_i$ (resp. $t_i$) is the number of chores in $\x_i$ with cost $b$ (resp. $a$) by $i$, the disutility of $i$ is $u_i = s_i  + t_i k$. Since $0\le s_i,t_i \le m$, the number of different disutility values $i$ can have in any allocation is at most $O(m^2)$. Thus, for any agent $i$, the number of times her disutility decreases is at most $O(m^2)$. Thus, an agent can become the big spender only $O(m^2)$ times. Hence, the number of identity changes of the BS is at most $O(m^2n)$.

Thus, the time needed to make the component of the big spender pEF1 is $O(m^3n^2)$. Since we need construct at most $n$ pEF1 components, the total time needed is $O(m^3n^3)$. In conclusion, Algorithm~\ref{alg:MakeInitGroups} terminates in $\poly{n,m}$ time.
\end{proof}

\spendingofLS*
\begin{proof} 
The spending of the LS cannot decrease during a price-rise step. Suppose $t$ is a chore transfer step and $s$ is the spending of the LS at $t$. If $t\le T$, transfers are from the BS $b$ to the LS $\ell$. Further, this happens only if $b$ pEF1-envies $\ell$. Thus, $b$ cannot spend less than $s$ after the chore transfer due to the pEF1 condition. If $t>T$, chore transfers are along alternating paths $b \rightarrow j' \rightarrow i \rightarrow j \rightarrow \ell$, where $j'\in \x^t_{b}\cap \mBB_i$ and $j\in \x^t_i\cap \mBB_\ell$, and $b$ pEF1-envies $\ell$. As discussed in Section~\ref{sec:aftert}, the spending of $i$ does not change during this chore transfer, hence $i$ cannot spend less than $s$ after the transfer. As before, $b$ cannot spend less than $s$ after the transfer due to the pEF1 condition. Thus, the spending of the least spender does not decrease in the run of Algorithm~\ref{alg:ef1po}.
\end{proof}

\compsarepEFone*
\begin{proof}
We prove this lemma in two parts: before $T$, and after $T$, the latter is proved in Lemma~\ref{lem:afterTinvariant} (iii). For the run of the algorithm prior to $T$, we proceed by induction. As the base case, we know from Lemma~\ref{lem:initcompprop} that the groups $\{N_r\}_{r\in[R]}$ are each pEF1 at the end of $\s{MakeInitGroups}$. As the inductive hypothesis, assume that each group is pEF1 at time-step $t$. After the event at $t$, a group can cease to be pEF1 potentially because of a (i) chore transfer, or (ii) price-rise. 

\begin{enumerate}
\item \textit{Chore transfer step}. From the algorithm it is clear that any chore transfer prior to $T$ happens from the big spender $b$ to $\ell$. Let $N_r$ be a group. We consider four cases:
\begin{itemize}
\item Both $b$ and $\ell$ lie outside $N_r$, then the bundles of agents in $N_r$ are unaffected, and the allocation continues to remain pEF1 after the transfer.
\item Both $b$ and $\ell$ lie in $N_r$. However, if this happens then the algorithm terminates with a pEF1 allocation, since $N_r$ is pEF1 at time-step $t$ by assumption.
\item $b$ lies in $N_r$ and $\ell$ does not. Let $i\neq b$ be an agent in $N_r$ and $(\x,\p)$ be the allocation at $t$. Since $N_r$ is pEF1, we have $\p(\x_i) \ge \p(\x_b \setminus j)$ for some $j\in \x_b$. Since $N_r$ is pEF1, $b$ does not pEF1-envy $i$ at $t$. Since $b$ loses a chore at $t$, $b$ does not pEF1-envy $i$ after the transfer. Further, since $b$ is the big spender in $H_r$, $\p(\x_i \setminus j') \le \p(\x_b \setminus j)$ for some $j'\in \x_i$. Thus, $\p(\x'_i \setminus j') \le \p(\x'_b)$, where $\x'$ is the allocation  after the transfer. Hence, $i$ does not pEF1-envy $b$ after the transfer. Thus, $N_r$ remains pEF1 after the transfer.
\item $\ell$ lies in $N_r$ and $b$ does not. Let $i\neq \ell$ be an agent in $N_r$ and $(\x,\p)$ be the allocation at $t$. Since $N_r$ is pEF1, we have $\p(\x_i \setminus j) \le \p(\x_\ell)$ for some $j\in \x_i$. Since $\ell$ gains a chore $j'$ at time $t$, the allocation $\x'$ after the event at $t$ satisfies $\p(\x'_\ell \setminus j') = \p(\x_\ell) \le \p(\x_i) = \p(\x'_i)$. Hence $\ell$ does not pEF1-envy $i$ after the transfer. Further, since $\ell$ gains an additional chore and $N_r$ is pEF1 at $t$, $i$ will not pEF1-envy $\ell$ after the transfer. Thus, $N_r$ remains pEF1 after the transfer.
\end{itemize}

\item \textit{Price-rise step}. Suppose $N_r$ undergoes a price-rise at $t$. At $t$, $N_r$ is pEF1 by assumption. Since the prices of all chores in $N_r$ are raised by the same factor, it continues to be pEF1 after the price-rise.
\end{enumerate}
In both cases, we see that $N_r$ continues to remain pEF1 after the event at $t$. We conclude by induction that each group remains pEF1 throughout the run of the algorithm.
\end{proof}

We denote by $(\x^{t},\p^t)$ the allocation and price vector at time-step $t$. 
\BSincompofformerLS*
\begin{proof}
Let $t$ be the first time-step when the big spender (BS) $b$ lies in a group $N_r$ which also contains an agent $i$, who was a least spender (LS) most recently at time-step $t'\le t$. If $t'= t$, then $N_r$ contains both the LS and the BS, in which case the allocation must be pEF1 due to Lemma~\ref{lem:compsarepEF1}. Thus, $t' < t$.

Suppose $N_r$ underwent a price-rise at a time-step $t''$ between $t'$ and $t$. Then the BS $h$ at $t''$ must be in $N_r$. If that happens, then at $t''$, the big spender lies in a group which also contains a former LS $i$. Since $t'' < t$, this contradicts the definition of $t$. Thus, $N_r$ does not undergo a price-rise between $t'$ and $t$. We next show that:
\begin{claim}\label{clm:BSnotformerLS}
$b$ cannot be a former least spender, unless the allocation is already pEF1.
\end{claim}
\begin{proof}
For sake of contradiction, assume $b$ was an LS most-recently at time-step $t_0 < t$. From the properties of the algorithm (Lemma~\ref{lem:algobs}), $b$ ceased to be an LS at $t_0$ because she received a chore $j$. Subsequent to this, $b$ cannot receive another chore, since $t_0$ is the most recent time-step when $b$ was an LS, and only an LS can receive a chore (Lemma~\ref{lem:algobs}). Thus we must have $\x_b^t \subseteq \x_b^{t_0}\cup \{j\}$. Further, $b\in H_r$ does not undergo price-rise between $t_0$ and $t$, otherwise it would contradict the definition of $t$. Hence we obtain: $\p^t(\x^t_b \setminus j) \le \p^{t_0}(\x^{t_0}_b)$. 

Let $\ell$ be the LS at $t$. From Lemma~\ref{lem:spendingofLS}, we know $\p^t(\x^t_\ell) \ge \p^{t_0}(\x^{t_0}_b)$, since $b$ was the LS at $t_0$. Putting it together, we get $\p^t(\x^t_\ell) \ge \p^t(\x^t_b\setminus j)$, i.e., $(\x,\p)$ is pEF1 at $t$, since $b$ is the BS at $t$.
\end{proof}
 
With Claim~\ref{clm:BSnotformerLS} in hand, we can now prove Lemma~\ref{lem:BSincompofformerLS}. Since $b$ is not a former LS, $b$ has not received a chore. Also, $b$ has not undergone a price-rise. Hence: $\p^t(\x^t_b\setminus j) \le \p^{t'}(\x^{t'}_b\setminus j') \le  \p^{t'}(\x^{t'}_i)$, for some $j\in \x^t_b$ and $j'\in \x^{t'}_b$, since $H_r$ is pEF1 at $t'$ by Lemma~\ref{lem:compsarepEF1}. Finally, by Lemma~\ref{lem:spendingofLS}, we get $\p^{t'}(\x^{t'}_i)\le \p^t(\x^t_\ell)$, where $\ell$ is the LS at $t$. Putting these together, we get: $\p^t(\x^t_\ell) \ge \p^t(\x^t_b\setminus j)$, showing that $(\x,\p)$ is pEF1 at $t$.
\end{proof}

\aboveHr*
\begin{proof}
Let $t \le T$ be the time-step at which $N_r$ was raised, i.e., the most recent price-rise. Let $h$ be the big spender at $t$. We know that $h\in H_r$. Let $i$ be an arbitrary agent in $N_{<r}$.

Since $i$ is in a raised group $N_{r'}$, where $r' < r$, $i$ cannot have gained a chore before $T$. Let $t'<t$ be the time at which $N_{r'}$ was raised. Then:
\begin{itemize}
    \item if $i$ gained a chore before $t'$, then $N_{r'}$ cannot have been undisturbed at $t'$.
    \item if $i$ gained a chore at time $t'' \in (t',T)$, then $i$ had to have been the LS at $t''$ (Lemma~\ref{lem:algobs}). Hence at $t''<T$, the LS was in a raised group. This contradicts the definition of $T$.
\end{itemize}

Hence $i$ can only lose chores. Just after $t'$, all chores of $i$ are priced at $k$. Hence the spending of $i$ after $t'$ is always a multiple of $k$. Further, until $t$, $h$ is undisturbed, and contains only chores of price 1. Let $c = \p^t(\x^t_h)$ be the number of (price 1) chores initially assigned to $h$.

Note that $N_{r'}$ must contain an agent who initially has at least $c$ price 1 chores, since $N_{r'}$ undergoes price-rise before $N_r$. Since each group is pEF1, $i$ must have at least $c-1$ price 1 chores. Hence just after $t'$, we have $\p^{t'}(\x^{t'}_i) \ge k(c-1)$. After $t'$, the spending of $i$ potentially drops because of transfer of chores of price $k$.

At $t$, since $h$ is the BS, we know: $\p^t(\x^t_i)-k \le \p^t(\x^t_h)-1$. Since the spending of $i$ is a multiple of $k$, we have:
\[\p^t(\x^t_i) \le k + k\cdot\left\lfloor \frac{\p^t(\x^t_h)-1}{k} \right\rfloor = k + k\cdot\left\lfloor \frac{c-1}{k} \right\rfloor\]

Just before the time before $t$ that $i$ lost her last chore, since $h$ is not the BS, the spending of $i$ has to be at least $k + (c-1)$. Hence the spending of $i$ at $t$ is $gk$ where:
\begin{itemize}
    \item if $k | (c-1)$, then $g\in \{ (c-1)/k, 1+(c-1)/k\} $
    \item otherwise, $g = 1+\lfloor (c-1)/k \rfloor$
\end{itemize}

Since $i$ is an arbitrary agent in $N_{<r}$, we see that every agent in $N_{<r}$ spends almost the same amount (up to $k$). Hence there is no pEF1-envy among agents of $N_{<r}$ at $t$. From $t$ to $T$, agents in $N_{<r}$ only lose chores, and the gap between the spendings of any two agents cannot ever exceed $k$. Hence at $T$, agents in $N_{<r}$ do not pEF1-envy each other. Since the LS is in $N_{<r}$ at $T$, we conclude that agents in $N_{<r}$ do not pEF1-envy any other agent.
\end{proof}

We see from the previous analysis that $\p^t(\x^t_i) \ge c-1$, where $c = \p^t(\x^t_h)$.

\belowHr*
\begin{proof}
Let $t \le T$ be the time-step at which $N_r$ was raised, i.e., the most recent price-rise. Let $i$ be an arbitrary agent in $N_{>r}$. Since $i$ is in an unraised group, $i$ cannot have undergone a price-rise or ever lost a chore. 

Suppose $t'$ is the last time $i$ gained a chore $j$. Then at $t'$, $i$ is a LS. Thus, $\x^{t'}_i \cup \{j\} = \x^T_i$, and  $\p^{T}(\x^{T}_i\setminus j) = \p^{t'}(\x^{t'}_i)$. Since the spending of the LS weakly increases, $\p^{t'}(\x^{t'}_i) \le \p^T(\x^T_{\ell})$, where $\ell$ is the LS at $T$. Hence we obtain: $\p^{T}(\x^{T}_i\setminus j) \le \p^T(\x^T_{\ell})$, showing that $i$ does not pEF1-envy the LS. 

Suppose $i$ did not ever gain a chore. Then, the bundle of $i$ has been undisturbed from the start. Since $h$ is the BS at $t$, we have $\p^T(\x^T_i\setminus j') = \p^t(\x^t_i\setminus j') \le c-1$, for some $j'\in \x^T_i$. We know that $\ell$ spends at least $c-1$ at $t$. We know $\ell$ does not gain a chore. We show that it is not possible for $\ell$ to lose a chore after $t$, unless the allocation is already pEF1.

Suppose $\ell$ loses a chore $j$ for the last time at $t'\in(t,T)$. At $t'$, $\ell$ was the BS. Since there are no price-rises after $t$, the spending of the BS up to the removal of one chore weakly decreases. Hence $\p^T(\x^T_b\setminus j') \le \p^{t'}(\x^{t'}_\ell\setminus j)$, where $b$ is the BS at $T$. Since $\ell$ doesn't lose a chore after $t'$, $\p^{t'}(\x^{t'}_\ell\setminus j) = \p^{T}(\x^{T}_\ell)$. Thus the allocation is already pEF1. 

If $\ell$ did not lose any chore after $t$, then $\p^{T}(\x^{T}_\ell) = \p^{t}(\x^{t}_\ell) \ge c-1$. Since $\p^T(\x^T_i\setminus j') \le c-1$, for some $j'\in \x^T_i$ if $i$ has not gained a chore, we see that in this case too $i$ does not pEF1-envy the LS at $T$.

Since $i$ was an arbitrary agent in $N_{>r}$, we conclude that no agent in $N_{>r}$ pEF1-envies any other agent.
\end{proof}

\afterTinvariant*
\begin{proof}
Recall that Lemmas~\ref{lem:aboveHr} and \ref{lem:belowHr} shows that (i) and (ii) hold just after $T$, and Lemma~\ref{lem:compsarepEF1} shows (iii) holds just after $T$. With this as our base case, we proceed by induction. Assume (i)-(iii) hold at some time step $t>T$. The event at $t$ can either be:
\begin{itemize}
\item A direct chore transfer from $b\in N_r$ to $\ell\in N_{>r}$. Such a transfer does not newly cause any agent to pEF1-envy $\ell$ since $\ell$ gains a chore, nor any agent to pEF1-envy $b$ since $b$ is the BS. Also, $b$ does not newly pEF1-envy any agent since she loses a chore, and neither does $\ell$ newly pEF1-envy any agent since she is the LS.
\item A chore transfer of the form $b \rightarrow j' \rightarrow i \rightarrow j \rightarrow \ell$, where $i\in N_{>r}$, $j'\in \x^t_b$, $j\in\x^t_i$, and $j$ is a chore that $\ell$ lost prior to $T$. Hence $p^t_j = p^t_{j'} = k$. Hence, $\p^{t+1}(\x^{t+1}_i) = \p^t(\x^t_i)$. Thus, this step effectively amounts to transferring a chore from $b$ to $\ell$. As argued in the previous paragraph, this does not cause any agent to newly start pEF1-envying another agent. 
\end{itemize}
Hence, by induction, (i)-(iii) hold after $t$ as well.
\end{proof}

\bibliography{references}
\end{document}